\documentclass[12pt]{article}
\usepackage{xcolor} 
\usepackage{amsthm,amsmath}
\usepackage{graphicx, psfrag,epsf}
\usepackage{ragged2e}
\usepackage{subcaption}
\usepackage{amssymb}
\usepackage[font=small,labelfont=bf]{caption}
\usepackage{natbib} 
\usepackage{footnote}

\usepackage{url} 
\usepackage{hyperref} 
\usepackage{breakurl}
\usepackage{tcolorbox}
\usepackage{enumitem}
\newlist{boxitemize}{itemize}{1}
\setlist[boxitemize,1]{label=\textbullet,leftmargin=2mm}

\newcommand{\blind}{0}

\newcommand{\hatmusc}[0]{\hat\mu_{SC}}

\numberwithin{equation}{section}
\theoremstyle{plain}
\newtheorem{thm}{Theorem}[section]

\newtheorem{principle}{Principle}[section]

\graphicspath{{Images/}}

\usepackage{times}

\begin{document}

\def\spacingset#1{\renewcommand{\baselinestretch}%
{#1}\small\normalsize} \spacingset{1}

\def\B0{\mbox{\boldmath $0$}}
\def\Ba{\mbox{\boldmath $A$}}
\def\Bb{\mbox{\boldmath $B$}}
\def\Bc{\mbox{\boldmath $C$}}
\def\Bd{\mbox{\boldmath $D$}}
\def\Bl{\mbox{\boldmath $L$}}
\def\Bn{\mbox{\boldmath $N$}}
\def\Br{\mbox{\boldmath $R$}}
\def\Bt{\mbox{\boldmath $T$}}
\def\Bv{\mbox{\boldmath $V$}}
\def\Bx{\mbox{\boldmath $X$}}
\def\By{\mbox{\boldmath $Y$}}
\def\Bz{\mbox{\boldmath $Z$}}

\def\btheta{\mbox{\boldmath $\theta$}}
\def\htheta{\mbox{$\hat{\theta}$}}

\def\balpha{\mbox{\boldmath $\alpha$}}
\def\bbeta{\mbox{\boldmath $\beta$}}
\def\bdelta{\mbox{\boldmath $\delta$}}
\def\bepsilon{\mbox{\boldmath $\epsilon$}}
\def\bGamma{\mbox{\boldmath $\Gamma$}}
\def\blambda{\mbox{\boldmath $\lambda$}}
\def\bLambda{\mbox{\boldmath $\Lambda$}}
\def\bmu{\mbox{\boldmath $\mu$}}
\def\bSigma{\mbox{\boldmath $\Sigma$}}
\def\btau{\mbox{\boldmath $\tau$}}
\def\btheta{\mbox{\boldmath $\theta$}}


\if0\blind
{
  \title{Principles of Conditionality and Layering of Error Rates with Application to Platform Trials}
  \author{
  Xinping Cui, University of California, Riverside\\
		and \\
  Emily Ouyang, University of California, Riverside \\
  and \\
  Yi Liu, Nektar Therapeutics\\
and \\
  Jingjing Yan Schneider, BeiGene\\
		and \\
  Hong Tian, BeiGene\\
		and \\
  Bushi Wang, IO Biotech\\
        and \\
  Jason C. Hsu, The Ohio State University}
  \maketitle
  Correspondence:  Xinping Cui (\texttt{xinping.cui@ucr.edu})

} \fi

\if1\blind
{
  \bigskip
  \bigskip
  \bigskip
  \begin{center}
    {\LARGE\bf Title}
\end{center}
  \medskip
  
} \fi

\begin{abstract}

There has been a misconception that only one type of error rate control is necessary in clinical trials, leading to debates over whether to prioritize Familywise Error Rate (FWER) or False Discovery Rate (FDR). This misconception has led to misleading statements about FWER control and proposals to shift towards FDR control, which could be manipulated by the industry.
    
    In reality, since the early 2000s, biopharmaceutical statistics have implicitly applied two layers of Type I error rate control. This aligns with Tukey's 1953 invention of Error Rate per Family (ERpF) for controlling error across studies, while FWER applies within each study. Our paper clarifies this layering, using Platform trials to demonstrate the verifiable conditions needed across studies for the FDA to fulfill its regulatory mission.
    
    We show that controlling FWER within a study at $5\%$ inherently controls ERpF across studies at 5-per-100, regardless of study correlations. This supports current regulatory practices that protect public health while fostering innovation. We also address concerns about ERpF stability in Platform trials, where shared controls introduce dependencies. By applying the Conditionality Principle and utilizing an innovative Shiny app, we explore how correlations impact ERpF variability, providing deeper insights for informed decision-making.
    
    Our findings, supported by principles like Layering of Error Rate Controls and the No-gaming Principle, are particularly relevant as Platform trials gain popularity for their efficiency in testing multiple treatments simultaneously.


	\end{abstract}

\noindent%

\spacingset{1.45} 



\section{Introduction}
Clinical trials are, and have long been, an integral part of biomedical and pharmaceutical research. But as clinical study designs continue to evolve, the complexity of these trials has increased dramatically. Modern frameworks, including adaptive designs, platform trials, and multistage testing, introduce additional layers of structure, dependence, and decisions. Many classical methods of clinical trial analysis were developed under simplifying assumptions that are no longer appropriate in these settings. When applied naively, such classical procedures can produce inferences that are logically inconsistent and may fail to maintain error rate control.

While these inferential difficulties are often attributed to  multiplicity or dependence, we argue that the fundamental root is a failure to condition on the experiment actually performed. As we will show in this paper,  for complex modern clinical trials such as platform trials, if a shared control group happens to under-perform (randomly shows poor outcomes), it lowers the bar for all drugs in that trial. While a 5$\%$ ``unconditional'' false approval rate (FAR) provides a necessary long-run guarantee for regulatory policy, the actual risk for this trial might be 20$\%$. We further propose to apply Conditionality Principle to such trials in order to ensures that statistical evidence remains persuasive by focusing on the specific data-generating process of the trial at hand, protecting against the "luck of the draw" that makes marginal averages misleading for individual clinical decisions. Based on this principle, we recommend that regulators and sponsors do not just accept the unconditional guarantee but also perform a conditional assessment. 



    \begin{principle}[The Conditionality Principle]\label{principle:Contitionality}
	If, among many experiments that can inform on an unknown parameter, one is randomly selected and performed, then the information available about the unknown parameter depends only on the experiment actually performed.  
\end{principle}
\noindent See Berger (\citeyear{Berger(1985)}) page 31 and Casella and Berger (\citeyear{Casella&Berger(1990)}) page 268 for details.

Separate from these inferential concerns, there is the regulatory goal of controlling False Approval Rate (FAR). To make explicit the multi-level nature of regulatory decision making, we propose the Principle of Layering of error rates, which emphasizes that error rates should be controlled at each layer of decision making which is detailed in Section \ref{Sec:methods}.

The remainder of the article is organized as follows. Section \ref{Sec:methods} defines the regulatory goal of controlling the False Approval Rate (FAR). We propose the Principle of Layering of error rates, which should be controlled at each layer of decision-making. We also prove that if a regulator controls the Error rate familywise (ERFw) at 5$\%$ within every individual study, the long-run Error Rate per Family (ERpF), i.e., FAR,  across all studies will also be 5-per-100 on average.  Section \ref{Sec:conditionality} demonstrates that a 5$\%$ "unconditional" FAR is not enough to ensure the stability of a single platform trial, and proposes a conditional assessment of FAR through applying the Conditionality Principle and the Confident Platform Random Sizer App that we developed. Section~\ref{Sec:FDR} discusses limitations of False Discovery Rate (FDR) in a regulatory setting.  We conclude with a discussion in Section \ref{Sec:discussion}.

\section{Layering of Error Rates}
\label{Sec:methods}

In this section, we describe the layered error-rate framework underlying current regulatory practice. While elements of this framework have appeared in prior work, we present it here in a unified form that distinguishes error control within studies from error control across studies. 


\subsection{Incorrect Approval Rate (IAR) 
in a regulatory setting}




For regulatory agencies like the European Medical Agency (EMA) and the United States Food and Drug Administration (FDA), the primary purpose of Type I error rate control is to limit the rate of approving a treatment lacking efficacy (e.g., drug), 
which we refer to as Incorrect Approval Rate (IAR) control. In the United States, the FDA currently permits at most five incorrect approvals per one hundred ineffective compounds studied, corresponding to an IAR of 5-per-100.  Rather than attempting to control this error rate directly across various study designs, the FDA instead mandates that each individual study controls its familywise error rate at  $5\%$ within each $\textit{study}$. In enforcing this mandate, control of IAR at $5\%$ across studies is also achieved. To understand what this means, it is important to define the different layers of error-rate control.

\subsubsection{Two layers of families and error rates}

In many modern inferential settings, when there are decisions to be made both \textit{within} and \textit{across} studies, two layers of families may exist. As such, there are two different error rates we need to consider: 

\begin{description}
	\item[Within study] Within each study, there may be multiple \textit{doses} and/or \textit{endpoints}.  
	The \textit{within-study} family of inferential statements is the collection of (prespecified) statements to be made about them (e.g. confidence intervals for low dose efficacy in the primary endpoint and high dose efficacy in the secondary endpoint).  
	\item[Across studies] 
 Across all studies, the outcome of each study is whether the study compound is approved or not. The \textit{across-study} family of inferential statements is the collection of statements to be made about whether the test compound of each study in the family has efficacy or not. 
 
\end{description}

These layers correspond to different stages at which information is revealed and decisions are made. 
For instance, \citet{DingEtAl(2018)} has a group of inferences (lower layer family) comprising four simultaneous confidence intervals for each single nucleotide polymorphism (SNP), assessing treatment efficacy across three possible genotypes. At the higher layer, there are a panel of inferences comprising confidence statements across all SNPs on a genome-wide scale, evaluating whether each SNP is ``predictive"  of treatment efficacy. 
As another example, \citet{HanTangLinHsu(2022)} has a lower layer family of three confidence intervals for each biomarker value, assessing  the efficacy of a new treatment versus a control within three distinct groups: the marker-positive $ g^+ $ subgroup, marker-negative $ g^- $ subgroup, the overall population $ \{g^+, g^-\} $. At the higher layer, there are confidence statements across uncountable many biomarker values.

In platform trials, within each regimen, there exists a family of efficacy assessment on primary and secondary outcomes. For example, the HEALEY ALS Platform Trial has the primary outcomes (e.g., ALSFRS-R total score disease progression) and secondary outcomes (e.g.,  respiratory function). Subsequently the regimens in the platform trial form a higher layer of family of studies, each of which either succeeds or fails in claiming efficacy of the compound compared to the placebo.

\begin{principle}[Principle of Layering of error rate controls]\label{principle:layering}
	When there are multiple layers of families of statistical inferences, error-rate control can be applied in layers, to control the incorrect decision rate at each layer as well as for the overall decision-making process.  
\end{principle}

Because decision-making occurs sequentially at multiple inferential levels—both within and across studies — separate error-rate controls are needed at each layer, necessitating layering of error rates.


\subsubsection{Tukey's three Multiple Comparison error rates}

Tukey (1953; reprinted in 1994) formalized multiple comparisons by framing inference in terms of confidence intervals and introducing several distinct error rates.

\begin{description}
	\item[error rate per comparison] $ = \frac{\mbox{number of erroneous statements}}{\mbox{number of comparisons}} $
	\item[error rate per family] $ = \frac{\mbox{number of erroneous statements}}{\mbox{number of families}} $
	\item[error rate familywise] $ = \frac{\mbox{number of erroneous families}}{\mbox{number of families}} $ 
\end{description}
In these definitions, each ``statement'' is a confidence interval, and an erroneous statement is a confidence interval that does not cover the true parameter value.  
An erroneous family is a family with at least one confidence interval not covering its true parameter value.

\subsubsection{Error Rate Familywise (ERFw) control guards 
regulatory Incorrect 
Approval Rate (IAR)}

Within a study (a clinical trial), an incorrect statement about a primary endpoint can lead to an incorrect approval of a compound, and an incorrect statement about a secondary endpoint can lead to an incorrect claim of additional efficacy on the drug label.  
Since any incorrect statement within a study can lead to a (regulatory) incorrect decision, this has led to controlling error rate \textbf{familywise} (ERFw) \textit{within each study} since the early 2000s. 

Within a platform trial, different treatment arms  may share a common control arm and therefore are not independent. This dependence is \textit{short-term}, meaning that studies are correlated within the groups of studies that share data or analysis infrastructure, while studies without such sharing are independent. Across the infinitely many studies in independent trials and in platform trials, the following theorem allows us to show if the ERFw of each study  is set at 5\% level, the aforementioned dependence does not inflate the regulatory IAR beyond 5-per-100 (or 5\%), as long as dependence is short-term (within each platform trial only, while studies in different platform trials remain
independent).

We now establish an unconditional guarantee for error-rate control across studies, demonstrating that layered error-rate control remains valid in the long run even under short-term dependence.

\begin{thm}[Law of Large Numbers for random variables with short term dependence]\label{thm:SLLNcorrelated}
	Let $ X_1 , X_2, \ldots, $ have means $ \mu_1, \mu_2, \ldots $ variances $ \sigma^2_1, \sigma^2_2, \ldots $ and covariance $ Cov\{X_i, X_j\} $ satisfying 
	\begin{equation}\label{eq:LocalDependence}
		Cov\{X_i, X_j\} \le \rho_{j-i} \sigma_i \sigma_j ~ (i \le j)
	\end{equation}
	where $ 0 \le \rho_{m} \le 1  $ for all $ m = 0, 1, \ldots $. 
	If the series $ \sum_{1}^\infty \rho_m $ and $ \sum_{1}^{\infty} \sigma^2_i (\log i)^2 / i^2 $ are both convergent, then as $ n \rightarrow \infty $
	\begin{equation}
		\frac{1}{n} \sum_{i=1}^{n} X_i - \frac{1}{n} \sum_{i=1}^{n} \mu_i \rightarrow 0 ~{w.p. 1}.
	\end{equation}
\end{thm}
\begin{proof}
	See Theorem E on page 27 of Serfling(\citeyear{Serfling(1980)}).
\end{proof}

Even though individual studies and those in platform trials are conducted at overlapping times, we can think of them as occurring in a (countably) infinite sequence of clinical trials, indexed as Study $ \ldots , -1, 0, 1, \ldots $.  
Consider a finite number of $ K $ studies, some of which are in platform trials while others are not.  
Within these $ K $ studies, some compounds have efficacy while others do not.  
Studies with efficacious compounds that fail, i.e. failing to reject the null hypothesis of no-efficacy, do not lead to incorrect regulatory approval. 
We thus isolate those studies with compounds that lack efficacy (and therefore should not gain regulatory approval) as studies with index in $ N $, and let $ n $ denote its cardinality (the number of studies with index in $ N $).

Let $ \ldots , Y_{-1}, Y_0, Y_{1}, \ldots $ be the indicator functions of past, present, and future studies that lead to efficacy claims from rejection of the null hypothesis that the compound studied lacks efficacy.  
Among studies with index in $ N $ that should not lead to regulatory approvals, we relabel their $ Y_i $ as $ X_i $, i.e., $ X_i $ is the indicator of a Type I error being committed as a result of Study $ i, ~ i \in N, $ a study with compound that lacks efficacy,
\begin{equation}
X_i =  \left\{ 
\begin{array}{ll}
	1 & \mbox{if a Type I error is made as a result of Study} ~i, i \in N,\\
	0 & \mbox{otherwise}.
\end{array}
\right. \label{eq.ErrorIndicator}
\end{equation}
Then $E(X_i)=\mu_i$ 
for all $ i \in N$. 
For $ X_i, X_j $ within the same platform trial, the correlation $ \rho_{j-i} $ in (\ref{eq:LocalDependence}) is between zero and one.  
But with short term dependence, for $ X_i, X_j $ in different platform trials, $ \rho_{j-i} = 0$.  
For example, if the practical maximum number of study arms within a single platform trial is 100,  then the $ \rho_{j-i} $ in (\ref{eq:LocalDependence}) would satisfy 

$$ \rho_{j-i} \left\{ 
\begin{array}{ll}
	\le 1 & \mbox{if } |j - i| \le 100\\
	= 0 & \mbox{if } |j - i| > 100 .
\end{array}
\right. $$
Let $ m $ = $ |j - i | $, then $ \rho_m = 0 , ~ m = 101, 102, \ldots $ for platform trials.  
Therefore the series $ \sum_{1}^\infty \rho_i $ is convergent.  
The variance convergence condition is also satisfied since $ \sigma^2_i = \mu_i (1-\mu_i) \le 1$ and $ \sum_{i=1}^{\infty} (\log i)^2 / i^2 $ is a convergent series. 

Let $V$ denote the number of studies incurring Type I errors in $N$, i.e., 
\begin{equation*}
V = \sum_{i \in N} X_i.   
\end{equation*} 
Assuming $P(X_i = 1) = 5\%$ for every $ i \in N$, the expectation of $ V $ can be expressed as the sum of the probabilities for each $X_i$, leading to:
\begin{eqnarray}\label{eq:Additive}
E[V] & = & \sum_{i \in N} P(X_i=1) = 5\% \times n.
\end{eqnarray}

This \textit{additive}-$ \alpha $ multiplicity adjustment (\ref{eq:Additive}) for controlling error rate of 5 per a family of n studies is an \emph{exact equality} regardless of the dependency structure among the $ X_i $ variables, whether they are independent, positively dependent, or negatively dependent. 
It should in no way be confused with the Bonferroni \emph{inequality} multiplicity adjustment for the purpose of controlling ERFw.


This same decisions-triggering logic applies to genome-wide association studies (GWAS) as well. For example, \citet{DingEtAl(2018)} applies ERFw control within each specific SNP as lower layer of inference. Across the thousands of SNPs in the GWAS panel ((higher layer), this framework controls the expected number of SNPs ($E[V]$) with at least one false confidence interval coverage . Crucially, because this uses an additive multiplicity adjustment, the error control remains exact regardless of whether SNPs have linkage disequilibrium (correlation) or not. This is identical to a platform trial where controlling the ERFw at 5$\%$ within each drug study automatically controls the long-run IAR at exactly 5-per-100.

Moreover, Theorem \ref{thm:SLLNcorrelated} assures that as the number of studies ($n$) in $N$ grows infinitely large, the proportion of studies incurring Type I errors ($V/n$) will converge to the probability of a Type I error in any single study (i.e., ERFw)—$5\%$ in this case. In other words, by \textit{additive}-$ \alpha $ multiplicity adjustment  (on line 360 of FDA(\citeyear{FDA(2023)})),
controlling the Type I error rate familywise (ERFw) of each study \textit{exactly} at 5\% leads to controlling the error rate per family (ERpF) at \textit{exactly} 5\% for a family of $ n $ studies with in-efficacious compounds (which we phrase as 5-per-100).  Note that throughout this paper, ERpF and   
More generally, regardless of whether studies in the families are independent or dependent, having an error rate per family \textit{no more than} $ 5\% \times n $ for a family of $ n $ studies means to expect \textit{no more than} one incorrect approval \textit{per} a \textit{family} of 20 studies, five incorrect approvals \textit{per} a \textit{family} of 100 studies, etc. Such precision in error control exemplifies the regulatory commitment to balancing risk and benefit for patients, challenging the perception suggested by Bai et al\citet{BaiEtAl(2020)}. 

We include this result to clarify while long-run regulatory guarantees remain valid, they do not by themselves ensure coherent inference when decisions are interpreted conditionally on the realized experiment.

\begin{principle}[The No-gaming Principle]\label{principle:No-gaming}
	A regulatory approval system should not allow sponsors to game it. 
\end{principle}

\noindent Principle 2.2 is a fundamental policy requirement for any regulatory approval system. For example, FDR control is not used to control ERFw within a clinical trial because (as already mentioned) it can be manipulated by sponsors.  
Similarly, if the requirement for ERpF control across studies were merely to control $ E[V] $ at  $ 5\% \times n$, then an individual sponsor attempting to gain approval for its compound $ {Rx}^1 $ can  game the system by conducting a Platform trial of its own (artificially) adding compounds $ {Rx}^i, i = 2,3,4 $ which are 
of no commercial value, controlling the ERFw of $ {Rx}^1 $ at 17\% while controlling the ERFw of each of $ {Rx}^i, i = 2,3,4 $ at 1\%, gaming the system.  
To adhere to the No-gaming Principle, industry should continue with the current practice of the additive-$ \alpha $ adjustment, controlling the ERFw of each study at $ \alpha $.

\section{Applying the Conditionality Principle in Practice}
\label{Sec:conditionality}






We have previously established that strong ERFw control within each study ensures valid unconditional control of the IAR 
across studies over time. However, in platform trials, regulators and stakeholders may also be concerned with behavior of the IAR within a given platform. 
This concern is reflected in FDA (2023) guidance. In section F (``Multiplicity") of chapter III (``Considerations On Design and Analysis,") of Food and Drug Administration (2023) guidance document,  it is explicitly stated that
\begin{quote}
the probability distribution for the number of type I errors should be considered both in evaluating a proposed design and analysis plan and in evaluating the persuasiveness of results.
\end{quote}
It emphasizes that regulators are concerned not only with the expected number of false approvals (i.e. the unconditional IAR), but also with the variability and distribution of IAR within a given design (i.e., conditional IAR).  A complete understanding of error behavior requires consideration of both the marginal (unconditional) and conditional perspectives, as together they characterize the joint distribution of IAR.

Motivated by this perspective, we examine platform-specific conditional IAR, calculated under realized design features such as the number of regimens   and the observed shared control outcome,  while preserving unconditional IAR control across an infinite sequence of studies.

\subsection{Platform Trials}


Platform trials represent an emerging clinical trial design that has gained increasing popularity due to their ability to evaluate multiple treatments simultaneously  within a unified framework. 
This design enhances resource efficiency and increases the likelihood that participants receive active treatments rather than placebos. However, the relative novelty and structural complexity of platform trials leave substantial room for further methodological development, particularly with respect to foundational statistical principles and error rate control in the presence of multiple, interdependent comparisons\citep{KoenigetAl(2023)}. 

Because statistical inference and regulatory decisions are made after the studies are conducted, platform trials provide a concrete setting in which the conditionality principle is not merely philosophical but operationally essential. In such trials, inference and regulatory decision-making proceed conditional on multiple layers of realized information, including the number of treatment arms, enrollment patterns, interim adaptations, and the performance of the shared control group.


A prominent example is the HEALEY ALS Platform Trial, a double-blind, placebo-controlled study evaluating multiple investigational therapies for amyotrophic lateral sclerosis (ALS). In this trial, participants are randomized across treatment regimens, each compared against a shared control group. The shared control structure induces statistical dependence across treatment comparisons and naturally give rise to conditioning event that are central to inference and error rate assessment. These features motivate a formal conditional assessment of IAR control within a platform trial, which we develop in the next section.

\subsection{Conditional assessment of IAR control}

Applying Principle \ref{principle:Contitionality}, inference within a platform trial should be interpreted conditional on realized experiment, including the observed shared outcome and the dependence structure induced by control sharing. In practice, the relevant condition events include the number of regimens, the observed shared control mean, and the correlation structure across treatment comparisons.

To motivate this conditional perspective, 
consider Fieller's confidence interval for the ratio of two Normal means\citep{Fieller(1954)}. Unconditionally, the interval achieves exact $(1-\alpha)$ coverage across infinitely many replications. However, the form of Fieller's confidence interval depends on the realized data; it may be finite, infinite  (the entire real line), the complement of an interval, or empty (see Casella and Berger (\citeyear{Casella&Berger(1990)}), page 459).
Conditional on the interval  being the entire real line, it conveys no information about the parameter; 
conditional on it being empty, the effective confidence is 0\%. These examples illustrate that while unconditional guarantees remain valid, they do not fully describe the evidential strength in the realized experiment.
A complete assessment therefore requires consideration of both the marginal (unconditional) properties and the conditional behavior given the observed data.

Platform trials exhibit this same phenomenon. While a 5\% ERFw within each study guarantees an unconditional IAR of 5\% across an infinite sequence of studies, the evidential strength and error behavior within a particular platform depend on the realized shared control outcome and study configuration. Accordingly, we recommend assessing the conditional IAR within each platform trial, given the observed shared control outcome and realized design features, to ensure that it does not deviate materially from its unconditional guarantee.

\subsection{Conditional Stability of IAR is desirable }

We now study the distribution of erroneous efficacy claims within a single platform trial under the shared control structure. Although unconditional IAR control constrains the expected error rate across repeated experimentation, conditional error behavior may vary across realized platform configurations. Our goal is to characterize this variability formally.

Consider $k$ studies in a platform trial. Without loss of generality, we define efficacy as $\theta_i > 0$ and consider one-sided inference.
Let
\[
L
=
\left\{
i \in \{1,\dots,k\} :
\theta_i = \mu_i - \mu_{SC} \le 0
\right\}
\]
denote the set of ineffective compounds and let $|L|=\ell$. $\mu_{i}$ and $\mu_{SC}$ are true mean of the new compound $i$ and the shared control, respectively. Define 
\[
Z_i = \mathbf{1}\{\text{the confidence interval for } \theta_i \text{ excludes } 0 \text{ in the efficacy direction}\}, \quad i = 1,\dots,k,
\]
so that for each $i \in L$, $Z_i = 1$ corresponds to an erroneous efficacy conclusion (an incorrect approval), 
The total number of incorrect approvals is therefore
\[
V^k = \sum_{i \in L} Z_i.
\]
The distribution of $V^k$ depends on the unknown configuration $L$. 

Since neither the identities nor the number of $\theta_i\le 0$ are known, the exact distribution of $V^k$ cannot be determined without additional assumptions. In what follows, we will obtain a bound $ V^{*}$ that does not depend on $L$, and discuss two specific concerns on variability of $ V^{*}$ associated with sharing a control in a platform trial and the need for conditioning in accordance with the Conditionality Principle.

\subsubsection{Effect of correlation induced by sharing control on variability of \texorpdfstring{$V^{*}$}{V*}}

Suppose each confidence interval is constructed from a pivotal statistic  
\begin{equation}\label{eq:Tstar}
	T^\star_i = \dfrac{\hat{\mu}_{i}-\hat{\mu}_{SC}-\theta_i}{SE_{\hat{\mu}_{i}-\hat{\mu}_{SC}}}	
\end{equation}
whose sampling distribution does not depend on the efficacy parameter $\theta_i$. Here $\hat{\mu}_{i}$ and $\hat{\mu}_{SC}$ denote the sample mean for new compound arm $i$ and the shared control arm, respectively.  For proof of concept, consider the setting where 
\[
\hat{\mu}_i \sim N(\mu_i, \frac{\sigma_i^2}{n_i}), 
\qquad
\hat{\mu}_{SC} \sim N(\mu_{SC}, \frac{\sigma_{SC}^2}{n_{SC}}),
\]
Assume further that $\hat{\mu}_i$ and $\hat{\mu}_{SC}$ are independent. Then the estimator of the treatment effect is
\[
\hat{\theta}_i = \hat{\mu}_i - \hat{\mu}_{SC},
\]
with variance
\[
Var(\hat{\theta}_i) = \frac{\sigma_i^2}{n_i} + \frac{\sigma_{SC}^2}{n_{SC}}.
\]
Consequently,
$T^\star_i,i=1,\cdots,k$ have marginal standard normal distribution, and for each pair $i\ne j$,  

\[
(T_i^\star, T_j^\star)
\sim
N_2\!\left(
\begin{pmatrix}0\\0\end{pmatrix},
\begin{pmatrix}
1 & \rho_{ij}\\
\rho_{ij} & 1
\end{pmatrix}
\right).
\]
This scenario is particularly applicable in therapeutic areas like Alzheimer's Disease (AD) and Type 2 Diabetes Mellitus (T2DM), where the primary efficacy endpoints\footnotetext{which excludes the score test version of the logRank test} are well-defined and standardized. 

Let
\[
Z_i^\star = \mathbf{1}\{T_i^\star > z_\alpha\}
\]
denote the indicator of non-coverage for the one-sided $(1-\alpha)$ confidence interval, 
where $z_\alpha$ is the upper $(1-\alpha)$ quantile of the standard normal distribution. 
Under this construction,
\[
P(Z_i^\star = 1) = \alpha.
\]
Define
\[
V^\star = \sum_{i=1}^k Z_i^\star
\]
as the total number of non-coverage events across new compound arms. For any $i \in L$, we have $\theta_i \le 0$, and whenever the confidence interval excludes $0$ in the efficacy direction (i.e., $Z_i = 1$), it necessarily does not cover the true value (i.e., $Z_i^* = 1$).  Hence,
\[
Z_i \le Z_i^* \quad \text{for all } i \in L,
\]
which implies
\[
V^k = \sum_{i \in L} Z_i 
\;\le\; \sum_{i \in L} Z_i^* 
\;\le\; \sum_{i=1}^k Z_i^*
\;=\; V^*.
\]
Therefore, instead of $ V^{k} $, we consider the distribution of $ V^{*}$, an upper bound on the number of compounds that were erroneously inferred to be efficacious.  

Note that $V^*$ counts all one-sided confidence interval noncoverage events in the following cases. Therefore $V^*$ is more conservative as a bound.

\medskip
\noindent
\textbf{Case 1: Misses for efficacious arms.}  
Suppose $i \notin L$, so $\theta_i > 0$.  
If the confidence interval fails to cover the true value $\theta_i$, then
\[
Z_i^* = 1.
\]
However, this event does not correspond to a false approval. Thus, it contributes to $V^*$ but not to $V^k$.

\medskip
\noindent
\textbf{Case 2: False approvals.}  
If $i \in L$ and the confidence interval lies entirely above zero, then
\[
Z_i^* = 1 \quad \text{and} \quad Z_i = 1,
\]
which contributes to both $V^*$ and $V^k$.

The variance of $\operatorname{Var}(V^\star)$ can be written as 
\[
\operatorname{Var}(V^\star)
=
\sum_{i=1}^k \operatorname{Var}(Z_i^\star)
+
2\sum_{1\le i<j\le k} \operatorname{Cov}(Z_i^\star, Z_j^\star),
\]
where
\[
\operatorname{Var}(Z_i^\star) = \alpha(1-\alpha).
\]
Thus the variance of $V^\star$ depends on the covariances $\operatorname{Cov}(Z_i^\star, Z_j^\star)$ for $i \ne j$,  and we have 
\begin{equation}
\operatorname{Cov}(Z_i^\star, Z_j^\star)
=
P(T_i^\star > z_\alpha,\; T_j^\star > z_\alpha)
-
\alpha^2 .
\end{equation}
In platform trials with a shared control arm, the correlation between the standardized treatment effects arises from the common control mean. Assuming independent new compound treatment arms and equal outcome variances, we can easily prove that 
\[
\rho_{ij}
=
\operatorname{Corr}(T_i^\star, T_j^\star)
=
\lambda_i \lambda_j,
\]
where
\[
\lambda_i =
\left(1+\frac{n_{SC}}{n_i}\right)^{-1/2},
\]
and $n_i$ and $n_{SC}$ denote the sample sizes of new compound treatment arm $i$ and the shared control arm, respectively.

Each subset $ L $ has its distinct $ \rho_{ij} $ matrix, and thus $\lambda$ matrices, depending not only on the ratio of sample sizes between the treatment arms with index in $ L $ and the shared control arm,  but also on other covariates like baseline measurement of each patient.  To assess the stability of error rate control for compounds 
apparently with no efficacy, we consider the observed new compound treatment outcomes as well as the covariates as known. Consequently, we compute $ \rho_{ij} = \lambda_i\lambda_j , i \ne j,$ among the compounds in $L$ that show no efficacy.  Such computations are feasible within model-based analytical frameworks.


Table 1 illustrates how correlation among treatment comparisons affects the variability of $V^\star$. As the common correlation $\rho_{ij}$ increases, $StdDev(V^\star)$ increases substantially and can exceed the mean $E(V^\star)$.
For example, when $k=10$ and $\rho_{ij}=0.5$, $StdDev(V^\star)$ increases to 1.16 compared with 0.69 under independence. Moreover, the variability grows rapidly with family size $k$. If $StdDev(V^\star) > E(V^\star)$, the realized number of erroneous claims within a given family may fluctuate widely around its expectation, raising concerns about the stability of ERpF control. While independent studies approach stability when the family size is around $k=20$, this stabilization
does not occur in platform trials due to the correlation induced by the shared control arm.

To facilitate the assessment of the impact of correlations, analogous to how internet commerce tracks \textit{cookies}, we propose maintaining an informational “cookie” for each treatment or regimen in a platform trial. Table \ref{tab:CSM} presents the Control Sharing Matrix (CSM) for the HEALEY ALS Platform Trial, inferred from the results available on clinicaltrials.gov. For each regimen, the diagonal elements represent the number of patients randomized to the active treatment. The off-diagonal elements indicate the number of patients shared in the control group, to which each active treatment is compared, for each pair of regimens. We can then derive Table \ref{tab:cor_mat}, which presents the correlation matrix. Note that as mentioned before, the pairwise correlation in Table \ref{tab:cor_mat} where $\lambda_i = (1+\frac{n_{SC}}{n_{i}})^{-1/2}$. These $\lambda$ values can be readily calculated using existing software, making it preferable for its ease of extension to more complex correlation structures\footnote{For example, when two additional regimens are added later on, we do not need to calculate a six by six correlation matrix, instead, we only need to calculate the pairwise correlation based on the corresponding shared control sample size for each pair of regimens.} For the HEALEY ALS Platform Trial, the approximated $\lambda = $ (0.674, 0.722, 0.671, 0.690). Based on equation (3.2)\footnote{Equation (3.2) assumes all the regimens have the same $n_{SC}$, which may not be the case in real life (e.g., the HEALEY ALS Platform Trial).  In practice, one just uses the best available approximation, and we use the $\lambda$ from the factor-analytic approximation.} we can obtain variance-covariance matrix for $Z^{*}$ (see Table \ref{tab:CorrelationsZ}). Note that the elements sum to 0.2974 for $Var(V^\star)$, resulting in $StdDev(V^\star)$ = 0.5454, which exceeds the ERpF of 0.2.

\subsubsection{Effect of sample mean of shared control on the distribution of conditional \texorpdfstring{$V^{*}$}{V*}}

From the perspective of the Conditionality Principle, inference should be
evaluated conditional on the realized outcome of shared control, since regulatory
decisions are made after this quantity is observed. Recall the standardized
statistic
\[
T_i^\star
=
\frac{\hat{\mu}_i - \hat{\mu}_{SC}-\theta_i}
{\sqrt{\frac{\sigma_i^2}{n_i} + \frac{\sigma_{SC}^2}{n_{SC}}}},
\]
for treatment arm $i$, where
\[
\hat{\mu}_i \sim N\!\left(\mu_i,\frac{\sigma_i^2}{n_i}\right), \quad  \hat{\mu}_{SC} \sim N\!\left(\mu_{SC},\frac{\sigma_{SC}^2}{n_{SC}}\right),
\]
and
\[
\theta_i=\mu_i-\mu_{SC}.
\]
Conditional on $\hat{\mu}_{SC}$, 
\[
\hat{\mu}_i-\hat{\mu}_{SC}
\sim
N\!\left(
\mu_i-\hat{\mu}_{SC},\frac{\sigma_i^2}{n_i}
\right), 
\]
and
\[
\hat{\mu}_i-\hat{\mu}_{SC}-\theta_i
=
\hat{\mu}_i-\hat{\mu}_{SC}-(\mu_i-\mu_{SC}),
\]
we have 
\[
E(\hat{\mu}_i-\hat{\mu}_{SC}-\theta_i\mid\hat{\mu}_{SC})
=
\mu_{SC}-\hat{\mu}_{SC}.
\]
Let
\[
m_i(\hat{\mu}_{SC})
=
\frac{\mu_{SC}-\hat{\mu}_{SC}}
{\sqrt{\frac{\sigma_i^2}{n_i}+\frac{\sigma_{SC}^2}{n_{SC}}}},
\qquad
v_i
=
\frac{\sigma_i^2/n_i}
{\frac{\sigma_i^2}{n_i}+\frac{\sigma_{SC}^2}{n_{SC}}},
\]
we have 
\[
T_i^\star \mid \hat{\mu}_{SC}
\sim
N\!\left(
m_i(\hat{\mu}_{SC}),
v_i
\right).
\]
Notably, the conditional mean of
$T_i^\star$ does not depend on the efficacy parameter $\theta_i$. A larger (or smaller) realized value of $\hat{\mu}_{SC}$ relative
to its expectation $\mu_{SC}$ makes it harder (or easier) for experimental
compounds to demonstrate efficacy, regardless of the true treatment effect.  In this paper, a shared control with
$\hat{\mu}_{SC} > \mu_{SC}$ is referred to as an \textit{over-achieving}
shared control, while $\hat{\mu}_{SC} < \mu_{SC}$ is termed an
\textit{under-performing} shared control.

Conditional on the observed shared control mean $\hat{\mu}_{SC}$, the conditional non-coverage probability for arm $i$ is
\[
p_i(\hat{\mu}_{SC})
:=
P(Z_i^\star=1\mid \hat{\mu}_{SC})
=
P(T_i^\star>z_\alpha\mid \hat{\mu}_{SC}).
\]
so that
\[
p_i(\hat{\mu}_{SC})
=
1-\Phi\!\left(
\frac{z_\alpha-m_i(\hat{\mu}_{SC})}{\sqrt{v_i}}
\right),
\]
where $\Phi(\cdot)$ denotes the cumulative distribution function of the
standard normal distribution.

Conditional on
$\hat{\mu}_{SC}$, the treatment-arm estimators are independent, and hence  $Z_1^\star,\ldots,Z_k^\star$ are conditionally
independent Bernoulli random variables with success probabilities
$p_1(\hat{\mu}_{SC}),\ldots,p_k(\hat{\mu}_{SC})$. Therefore,
\[
V^\star \mid \hat{\mu}_{SC}
\]
follows a Poisson-binomial distribution in general. If, in addition,
\[
\sigma_i=\sigma_{SC}=\sigma,
\qquad
n_i=n_{SC}=n,
\qquad i=1,2,\ldots,k,
\]
then
\[
p_1(\hat{\mu}_{SC})=\cdots=p_k(\hat{\mu}_{SC})=:p(\hat{\mu}_{SC}),
\]
and therefore
\[
V^\star \mid \hat{\mu}_{SC} \sim \mathrm{Binomial}\bigl(k, p(\hat{\mu}_{SC})\bigr),
\]
where
\[
p(\hat{\mu}_{SC})
=
1-\Phi\!\left(
z_\alpha\sqrt{2}-\frac{\mu_{SC}-\hat{\mu}_{SC}}{\sigma}\sqrt{n}
\right).
\]

Conditional on an \textit{over-achieving} $\hat{\mu}_{SC}$,  the probability of 
$V^\star|\hat{\mu}_{SC}$ being towards lower numbers of erroneous approvals increases with $\hat{\mu}_{SC}$. In other words, the higher $\hat{\mu}_{SC}$ becomes, the more challenging it is for treatment arms to receive approval, presenting significant concerns for sponsors seeking treatment approvals. However, for regulators, especially when the treatment arms are non-efficacious, this situation is less problematic as it reduces the likelihood of erroneously approving a treatment.


Conditional on an \textit{under-performing} $\hat{\mu}_{SC}$, the probability of $V^\star|\hat{\mu}_{SC}$ being towards high numbers of erroneous approvals increases with decreasing $\hat{\mu}_{SC}$. In extreme cases, if $\hat{\mu}_{SC}$ is sufficiently low, almost all treatment arms may be erroneously approved. While sponsors might favor such increased ease of approval,   it is undesirable from a patients and regulatory perspective because it heightens the risk of erroneous approvals.

For illustration, consider a platform trial with $k=5$ treatment arms,
each tested at a one-sided significance level $\alpha=0.05$, yielding
$\text{ERpF}=0.05\times5=0.25$. The total sample size is 1200 with a
$5{:}1$ allocation ratio of treatment to control, which induces an
unconditional correlation of approximately $0.5$ between the estimated
treatment–control differences. For simplicity, we assume the true shared
control mean $\mu_{SC}=0$ and the treatment arm means $\mu_i=0$ for $i=1,2,3,4, 5$.
Figure~1 illustrates the conditional distribution of $V^\star$
for several realized values of $\hat{\mu}_{SC}$.
Figure 1a indicates an essentially negligible change in the probability distribution of $V^{*}|\hat{\mu}_{SC}(\text{over-achieving})$.
However, as shown in Figures 1b and 1c, the likelihood of $V^{*}=1|\hat{\mu}_{SC}(\text{under-performing})$ and $V^{*}=2|\hat{\mu}_{SC}(\text{under-performing})$ increases as $\hat{\mu}_{SC}(\text{under-performing})$ decreases, thereby changing the distribution of $V^{*}|\hat{\mu}_{SC}$. These results highlight the importance of examining the full conditional distribution of $V^{*}|\hat{\mu}_{SC}$, rather than relying solely on unconditional error rate or summary statistics.

\begin{figure}[h]
\centering


\includegraphics[width = \textwidth]{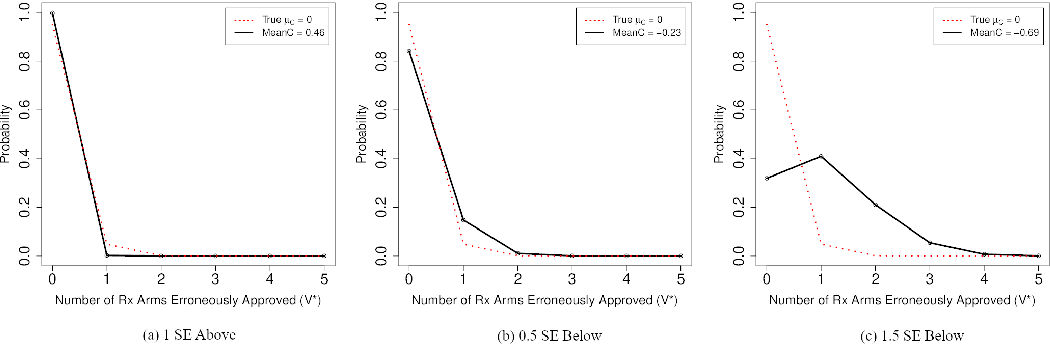}
\label{fig:pdf_vstar}
    \caption{ 
     \textbf{Probability distribution of $V^\star|\hat{\mu}_{SC}$} under a plotform trial with 5 treatment arms, each tested at $\alpha=0.05$, yielding an ERpF of  0.25. The total sample size is 1200, with a 5:1 treatment to control allocation, resulting in an unconditional correlation of 0.5 among treatment-control contrasts. For illustration purpose, the true shared control mean and treatment means are assumed to be $\mu_{SC}=\mu_i=0$. Panels (a)-(c) display the conditional distribution of $V^{*}$ given different values of $\hat{\mu}_{SC}$.}
    \label{fig:extreme}
\end{figure}

\subsubsection{A Confident Platform Random Sizer (Interactive visualization via Shiny App)}

To facilitate exploration of this conditional distribution, we developed an interactive visualization tool\footnote{Can be found at \url{https://eouya001.shinyapps.io/ConfidentPlatformRandomSizer/}} that implements the Conditionality Principle by displaying the entire distribution of $V^\star|\hat{\mu}_{SC}$, in line with  the FDA draft guidance. Our interactive plot allows users to examine how changes in design parameters affect the entire distribution, highlighting the impact of both over-achieving and under-performing $\hat{\mu}_{SC}$. Users can adjust parameter values interactively and observe in real time how the probabilities of erroneous approvals vary across different realization of $\hat{\mu}_{SC}$, thereby providing a practical way to explore the implications of conditional inference in platform trials.


 
In a simulated case study testing the efficacy of five drugs, the total resources allows for the participation of 600 people.
With this setup, users can adjust the input values using the app's sliders. Notably, ALS has a population standard deviation of 6.5 ($\sigma$ = 6.5), which the app uses internally to calculate standard errors of $\hatmusc, \hat \mu_i$ and $\hat \mu_i - \hatmusc$, $i = 1, \ldots, k
$,
where $k$ is the number of treatment arms. 
A scenario of particular concern to sponsors is detailed in \textbf{App: User Slider Controls} (Figure \ref{fig:Cui_Figure2})

    By selecting the sample size ratio, the app automatically calculates the unconditional correlation between the new treatment arm and control estimates. This feature is particularly beneficial to regulators, who have a vested interest in managing this aspect to stabilize the distribution of $V^{*}$. A lower correlation achieved by reducing the sample size ratio also decreases  the variation of $V^{*}$.
Therefore, the app can be instrumental in designing new studies. If regulators set a required correlation threshold, sponsors can use the app to determine the optimal sample size ratio randomization that meets this criterion. Based on the sample size allocation, the distribution of $V^{*}$ conditional on the sample mean of the shared control can also be evaluated (Figure \ref{fig:Cui_Figure3}).



Within a platform trial, the randomization ratio $ n_{Rx}: n_C $ of the new compound to the shared control can vary from a low of $ 1:1 $ - resulting in a sample size ratio of the shared control to each new compound of $ k:1 $—to a high of $ k:1 $, where the ratio reverses to $1:1$. Given a fixed total sample size for the platform trial, selecting the appropriate randomization ratio involves multiple factors. As shown in Chandereng et al (\citeyear{ChanderengEtAl(2020)}), if minimizing the variance of each individual treatment vs. control $ Rx:C $ comparison were the sole consideration, then as described in Appendix A of FDA (\citeyear{FDA(2023)}), one would set $ n_{Rx}: n_C $ to be essentially $\sqrt{k}$.  
This ratio can be a starting point for all the stakeholders to explore the choice of $ n_{Rx}: n_C $.  

Patients generally prefer a higher $n_{Rx}: n_C$ ratio, as it increases their chances of receiving a new treatment compound. For individual drug sponsors, the advantage of participating in a platform trial includes being able to pool control patients, which allows for a higher randomization ratio than the standard $ 1:1 $ typically seen in stand-alone studies. This higher ratio can facilitate more accurate effect estimations for new compounds.

However, the perspective of the master protocol sponsor tends to favor a lower 
ratio, which increases the sample size of the shared control, reducing the likelihood of an over-achieving control sample that complicates the demonstration of efficacy across new compounds.

From a regulatory standpoint, the stability of the IAR control is influenced by both the shared control's sample means and the correlation among the pivotal statistics $T^{*}$.  A lower ratio reduces the correlation among pivotal statistics, thereby enhancing the stability of IRA control.

In the HEALEY ALS Platform Trial, none of the four drugs tested in Regimens ABCD demonstrated efficacy in the primary endpoint, ALSFRS-R\footnote{Active treatments for Regimen A were Zilucoplan, B was Verdiperstat, C was CNM-Au8, and D was Pridopidine.}. The lack of efficacy could be attributed to either the inherent ineffectiveness of the compounds, an over-achieving shared control, or a combination of both. Conversely, as illustrated in App: Regulators Concerns, with an under-performing shared control sample mean  
there might have been a higher chance of the drugs being approved too readily, as exemplified by a $\hat{\mu}_{SC}$ that is 1.5 standard deviations below the $\mu_{SC}$. Such a scenario would raise significant concerns for regulators (Figure \ref{fig:Cui_Figure4}).


In sum, choice of the randomization ratio $ n_{Rx}: n_C $ should take all these into account.  Our app lets one find the sweet spot.

     \begin{tcolorbox}[title = App: User Slider Controls, 
        every float = \centering]
        \begin{minipage}[t]{0.57\linewidth}
        \vspace*{0pt}
        \begin{boxitemize}
            \item \textbf{$\alpha$ of Each Rx Arm}: 
            ERFw \textit{$\boldsymbol\alpha$} for each treatment arm. Note that ERpF is 
            $k\alpha$ \\
            \textit{$\alpha = 0.05$ for each treatment arm.}
            \item \textbf{$\boldsymbol{\delta:\mu_{SC} \pm \delta\sigma_{SC}}$:} 
            how many standard deviations is $\hatmusc$ away from its true mean.
            
            \textit{An example of 0.5 SD away from $\mu_{SC}$ 
(an over-achieving $\hatmusc$)
        }

            \item \textbf{Number of Non-Efficacious Rx Arms}: 
            $\ell$, assuming all treatment arms $k$ are non-efficacious, $l=k$.
            
            \textit{An example of $l=k=5$}

            \item \textbf{Total Sample Size}: sample sizes including all treatment arms and the shared control\\
            \textit{Total sample size = 600}
            \item \textbf{Correlation}: $\rho_{ij}$, Correlation between $T_i^\star$ and $T_j^\star$\\
    \textit{An example of $\rho_{ij} = 0.5$}   

    \item \textbf{Rx Arm: Control Sample Size Ratio}: Treatment to control sample size ratio. Note that the app allows for one to go higher, but we advise not to go higher than $k:1$. \\
    \textit{An example of 5:1 treatment-to-control sample size ratio}

        \end{boxitemize}
        \end{minipage}\hfill%
        \begin{minipage}[t]{0.4\textwidth}
        \vspace*{0pt}
        \vspace{1.25cm}

            \fbox{\includegraphics[width=\linewidth]{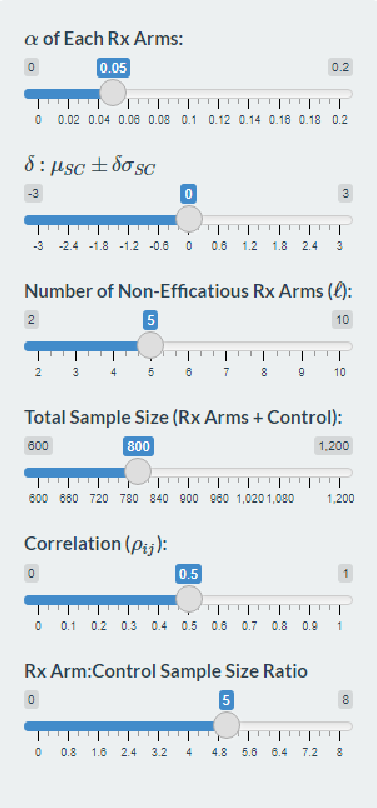}   }
           
            \captionof{figure}{User Slider Controls}
            
 \label{fig:Cui_Figure2}

    \end{minipage}

    \end{tcolorbox}

    \begin{tcolorbox}[title = App: Output,
    every float = \centering]
        \begin{minipage}[t]{0.42\linewidth}
        \vspace*{0pt}
        \begin{boxitemize}
            \item \textbf{Distribution of $V^{*}
            $}: Gives the expected value, standard deviation, and the probability distribution of $V^{*}$, the number of treatment arms erroneously approved. \\
            \textit{With an over-achieving $\hatmusc$, we see that the probability of erroneously approving a treatment arm is low}
            \item \textbf{Sample Allocation}: Gives the recommended sample sizes for each treatment arm and the shared control (rounded to nearest whole number)\\
            \textit{The study utilizes as 500:100 ratio, so a ratio of 5.}
        \end{boxitemize}
    \end{minipage}\hfill%
    \begin{minipage}[t]{0.55\linewidth}
    \vspace*{0pt}
        \fbox{\includegraphics[width=\textwidth]{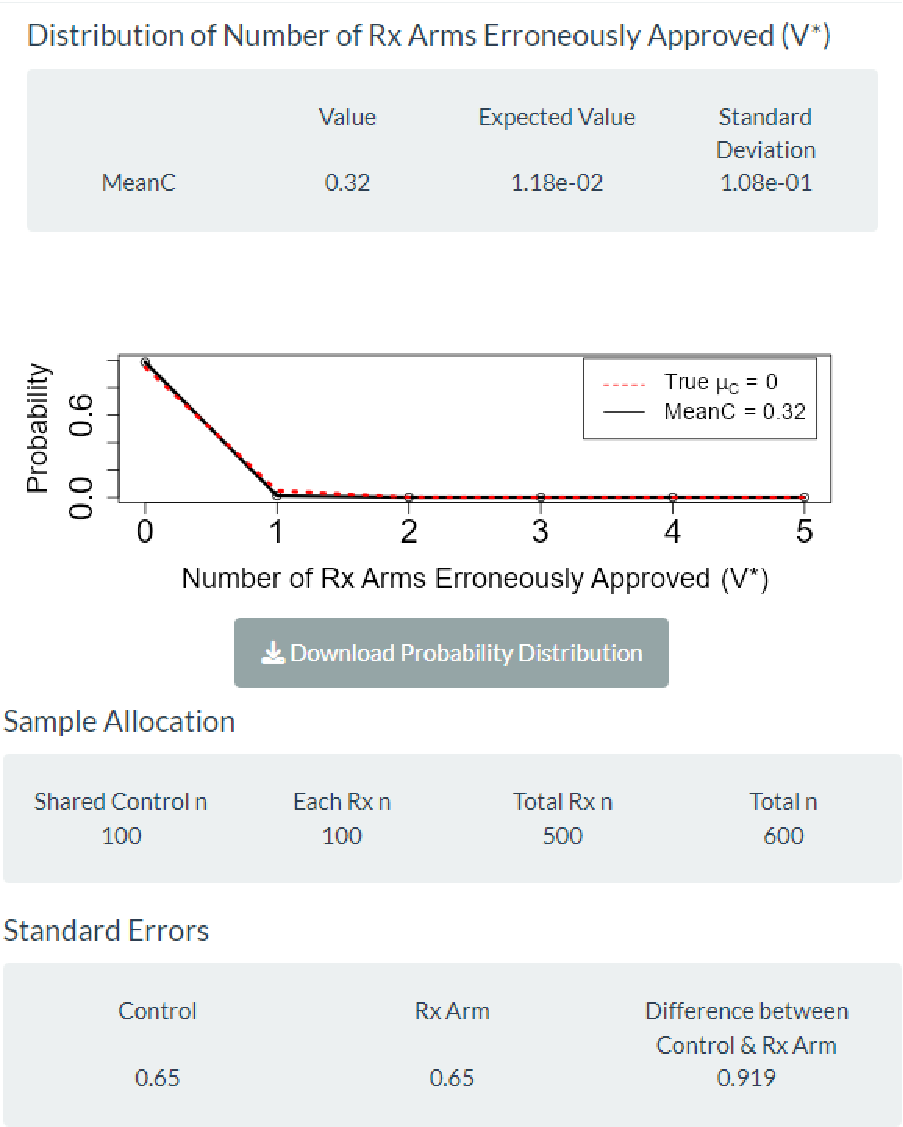}}
        \label{fig:Cui_Figure3}
        \captionof{figure}{App Output}
         \label{fig:Cui_Figure3}
    \end{minipage}
    \begin{boxitemize}
        \item \textbf{Standard Errors}: Gives standard errors for estimated control mean, treatment arm means, and the difference between the sample mean of the shared control and that of the treatment arm\\
        \textit{Standard errors were calculated based on the standard deviation $\sigma$ = 6.5}
    \end{boxitemize}

    \end{tcolorbox}

    \begin{tcolorbox}[title = App: Regulators Concerns, every float = \centering]
        We now present a scenario using our app where sponsors, such as MASS General, need not be concerned, but which should capture the attention of regulators.  The following inputs and outputs illustrate a hypothetical situation:
        
        \begin{minipage}[t]{0.42\linewidth}
        \vspace*{0pt}
        \vspace{1cm}
        \centering
        \textbf{Slider Inputs}
        \justifying
        \begin{boxitemize}
            \item \textbf{$\boldsymbol\alpha$ for Each Rx Arm}: $ 0.05$

            \item \textbf{$\boldsymbol\delta$}: $-1.5$

            \item \textbf{Number of Non-Efficacious Rx Arms}: $\ell = 5$

            \item \textbf{Total Sample Size}: 600

            \item \textbf{Correlation $\boldsymbol{\rho_{ij}}$}: 0.5

            \item \textbf{Sample Size Ratio}: 5:1\\
        \end{boxitemize}

    \end{minipage}\hfill%
    \begin{minipage}[t]{0.55\linewidth}
    \vspace*{0pt}
        \fbox{\includegraphics[width=\textwidth]{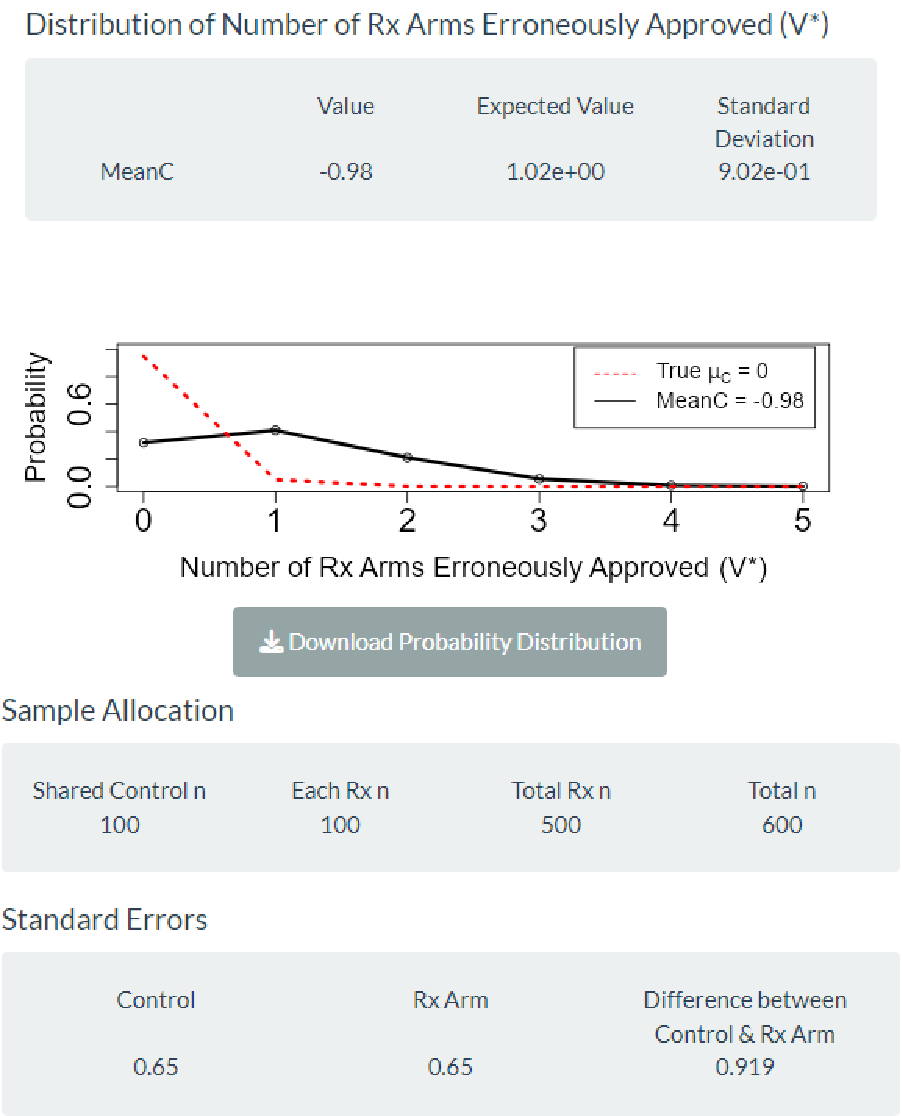}}
          \captionof{figure}{App Output}
              \label{fig:Cui_Figure4}
    \end{minipage}
    \vspace{1em}

    The probability distribution of $V^{*}$
  indicates that an under-performing $\hatmusc$ significantly increases the likelihood of erroneously approving at least one treatment, compared to its true efficacy. Regulators should be particularly concerned about this, as it suggests a non-negligible probability that ineffective drugs could mistakenly reach the market.
  
    \end{tcolorbox}
\section{False Discovery Rate (FDR) in a regulatory setting}
\label{Sec:FDR}

\subsection{FDR would change policy}

False Discovery Rate (FDR) control has been proposed as a method for controlling the Type I error rate across studies in Platform trials \citep{RobertsonEtAl(2023)}.  
FDR is defined as the expected proportion of incorrect rejections (rejecting a true null hypothesis) among all rejections, both 
correct and incorrect\footnote{Rejecting a false null hypothesis is considered ``correct'', while rejecting a true null hypothesis is ``incorrect''.}. 

Let $ H_{0i}, i = 1, \ldots, k $ 
represent the null hypotheses that each compound in a Platform trial with $ k $  studies lacks efficacy, then 
\begin{eqnarray}
	\mathrm{FDR} 
	& = &
	\mathrm{E} \left [ \frac{ \mbox{number of rejected } H_{0i} \mbox{ that are true}}
	{\mbox{total number of  rejected } H_{0i}} \right ] \label{eq:FDR} \\
	& = &
	\mathrm{E} \left [ \frac{ \mbox{number of in-efficacious compounds incorrectly approved }}
	{\mbox{total number of compounds approved, both efficacious and in-efficacious}} \right ]. \nonumber		
\end{eqnarray}
While current regulatory practice of controlling IAR (ERpF) across studies limits the number of in-efficacious compounds incorrectly approved to 5-per-100, 
controlling FDR in Platform trials 
would increase that number to beyond 5-per-100 as we will show below.

Recall $ L $ represent the subset of index $ \{1, \ldots, k\} $ for compounds that lack efficacy and $|L|=l$.  For each $ i \in L $,  $ Z_i $ is the indicator that the $ i^{th} $ confidence interval erroneously shows efficacy with $P(Z_i=1)$, and 
 \begin{equation*}
 V^k= \sum_{i \in L} Z_i.
 \end{equation*}
Then
\begin{eqnarray}
\dfrac{\mathrm{EFpF}}{l} & = & \dfrac{E(V^k)}{l}\\
&=&\dfrac{\sum_{i = 1}^{l} 
P\left \{ Z_i=1 \right \} }{l}\\
&=&\sum_{v = 1}^{l} \dfrac{v}{l} P\{ V^k = v \} 
 \label{eq.ERpFpercent} 
\end{eqnarray}


Now suppose a Platform trial has $ n^{\mathrm{false}} $ extremely false no-efficacy nulls (that will almost surely get rejected), then 
\begin{eqnarray}
	\mathrm{FDR} & = & \sum_{v = 1}^{l} \dfrac{v}{v+ n^{\mathrm{false}} } P\{ V^k = v \} \label{eq.FDRdiscount} \\
	& \le & \sum_{v = 1}^{l} \dfrac{v}{0 + n^{\mathrm{false}} } P\{ V^k= v \}. \label{eq.FDRupper}
\end{eqnarray}
How ERpF and FDR count the number of incorrect approvals $ v $ 
differently is evident by comparing (\ref{eq.FDRdiscount}) with (\ref{eq.ERpFpercent}).  
FDR attenuates the number of incorrect rejection $ v $ by the number of correct rejection of false nulls which we denote by $R$.  

From (\ref{eq.FDRupper}), we see $ n^{\mathrm{false}} $ being not small (relative to $ v $) allows inflation of 
$P \left \{ Z_i=1 \right\} $ 
while keeping FDR small.  
For simplicity, suppose each Platform trial has $ k = 2 $ studies only, and suppose the compound in Study $ w $ is highly efficacious but the compound in Study $ v $ is in-efficacious: 
\begin{itemize}
	\item $ H_{0w}^{{\scriptsize \mbox{false}}} $ is so false that $ P\left \{ \mbox{Reject } H_{0w}^{{\scriptsize \mbox{false}}} \right \} = 
	1 $; 
	\item $ H_{0v}^{{\scriptsize \mbox{true}}} $ is true so rejecting it constitutes the only incorrect rejection. 	
\end{itemize}
In this scenario, $ H_{0v}^{{\scriptsize \mbox{true}}} $ can be tested at a Type I error rate (ERFw) of $ 10\% $ (within Study $ v $) while keeping FDR in this Platform trial of a family of 2 studies at $ 5\% $, because 
\begin{eqnarray}
	FDR & = & E \left [ \frac{V}{R} \right ] \nonumber \\
	& = & \left ( \dfrac{1}{1 + 1} \right ) \times P\left \{ \mbox{Reject } H_{0v}^{{\scriptsize \mbox{true}}} \right \} \\
	& = & \frac{1}{2} \times 10\%  = 5\% \nonumber
\end{eqnarray}
So, if there is a sequence of many Platform trials with $ k = 2 $ studies following this scenario, then by the LLN (\ref{thm:SLLNcorrelated}), the Incorrect Approval Rate (IAR) is 10-per-100 (not 5-per-100).  



Adopting FDR control would represent a significant policy shift. Under this system, for every efficacious compound approved, regulators would essentially reward the industry with an \textit{additional} fungible False Approval Rate (FAR) token\footnote{We describe the tokens as fungible because any company can use them to seek approval for any compound.} that could be used toward the approval of inefficacious compounds. It is important to note that we find no explicit language in the Food, Drug, $\&$ Cosmetic Act or the Kefauver-Harris Amendment suggesting that the number of efficacious compounds correctly approved affects the guarding of the number of in-efficacious compounds incorrectly approved.

\subsection{Integrity of error rates against manipulation is desirable}

Whereas the IAR is only concerned with the rate at which in-efficacious compounds 
are incorrectly approved, clearly the issue with FDR is its denominator contains both correct and incorrect approvals of compounds.  Similar to the concern in Section 6 of Finner and Roter(\citeyear{Finner&Roter(2001)}), we wonder whether that might  subject FDR to manipulation as follows.  

If FDR control were implemented across platform trials involving multiple drug companies. It seems plausible that a strategic company might hesitate to submit their best compound, preferring instead to let other companies take that risk. Similarly, in platform trials with a single sponsor, there could be a temptation to include highly efficacious compounds that have low commercial value. This strategy could potentially be used to facilitate the approval of less efficacious but more profitable drugs.


 


\subsection{Struggles of FDR Not Restricted to Platform Trials}


FDR control also encounters serious limitations in other settings. In large-scale Genome-Wide Association Studies (GWAS), each hypothesis typically tests whether a SNP has no association with the disease outcome. Classical FDR methods assume that most hypotheses are truly null because then controlling FDR limits the number of erroneous findings among the reported results.
However, as demonstrated in \citet{DingEtAl(2018)} and further elaborated in \citet{DingEtAl(2021)}, due to linkage disequilibrium (LD), SNPs in close proximity are highly correlated. Consequently, when a single causal SNP influences a biological pathway, many neighboring SNPs may also appear statistically significantly associated with the disease outcome. In such settings, the complete null hypothesis may be false for a large proportion of SNPs, leading to a so-called “zero-null” scenario,  
and further resulting in small number of false discoveries $V$ relative to the total number of rejections $R$, and therefore
\[
\mathrm{FDR} = \mathbb{E}\left[\frac{V}{R}\right] \approx 0
\] automatically.
Obviously, controlling FDR doesn't tell us much about the quality of discoveries and a coherent alternative was proposed as a per-panel error rate \citep{DingEtAl(2018),DingEtAl(2021)}.

\subsubsection{Crosstalk plots: A Visual Tool for GWAS SNP Behavior}

To further illustrate why FDR becomes difficult to interpret in GWAS, we introduce crosstalk plots as a geometric representation of SNP dependence. To construct a crosstalk plot, select two SNPs, designating one as the \emph{conditioning} SNP, $C$ and the other as the \emph{response} or \emph{target} SNP, $T$. This designation is purely analytical and does not imply causality. The sample is partitioned according to the genotype of $C$, where $C \in \{0,1,2\}$ denote the number of minor alleles i.e., homozygous major, heterozygous, and homozygous minor. These three groups form the layers of a crosstalk plot. 

Within each stratum of $C$, we compute the conditional genotype distribution of the three genotypes of $T$. Each layer has a probability vector that sums to one, allowing it to be represented as a barycentric point on one of the three planes corresponding to the genotype strata of $C$. Connecting the three barycentric points across layers yields a trajectory whose geometry reflects dependence between the two SNPs. Figure \ref{fig:toyside} and \ref{fig:toytop} provide toy examples illustrating how the plot can demonstrate independence and dependence appearance and can be translated into a 2D projection form, similar to that seen in \citet{DingEtAl(2021)}. 

\begin{figure}[h]
    \begin{subfigure}{0.4\linewidth}
        \centering
        \includegraphics[width=\linewidth]{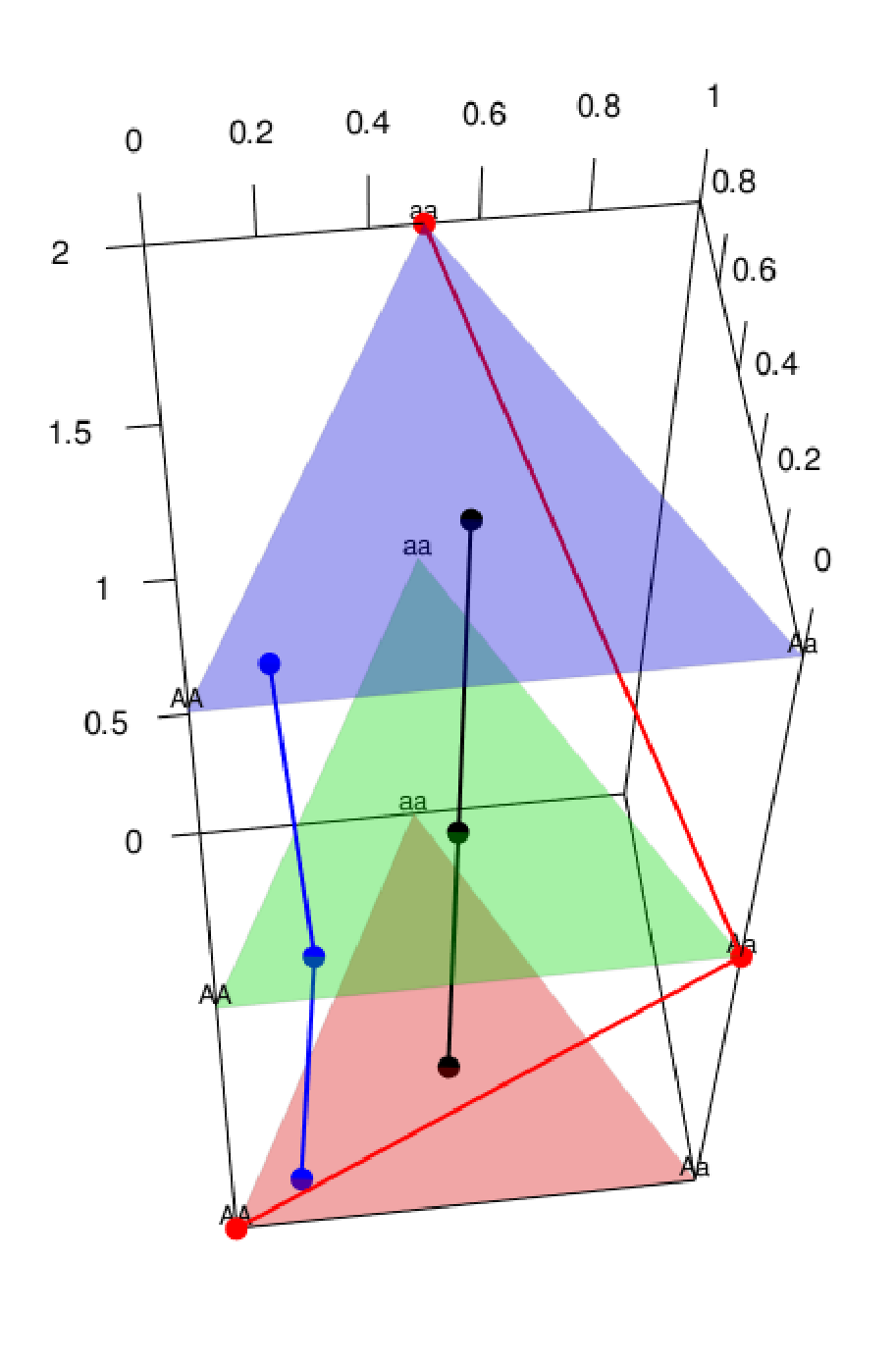}
        \caption{Side view}
        \label{fig:toyside}
    \end{subfigure}%
    \begin{subfigure}{0.6\linewidth}
    \centering
        \includegraphics[width=\linewidth]{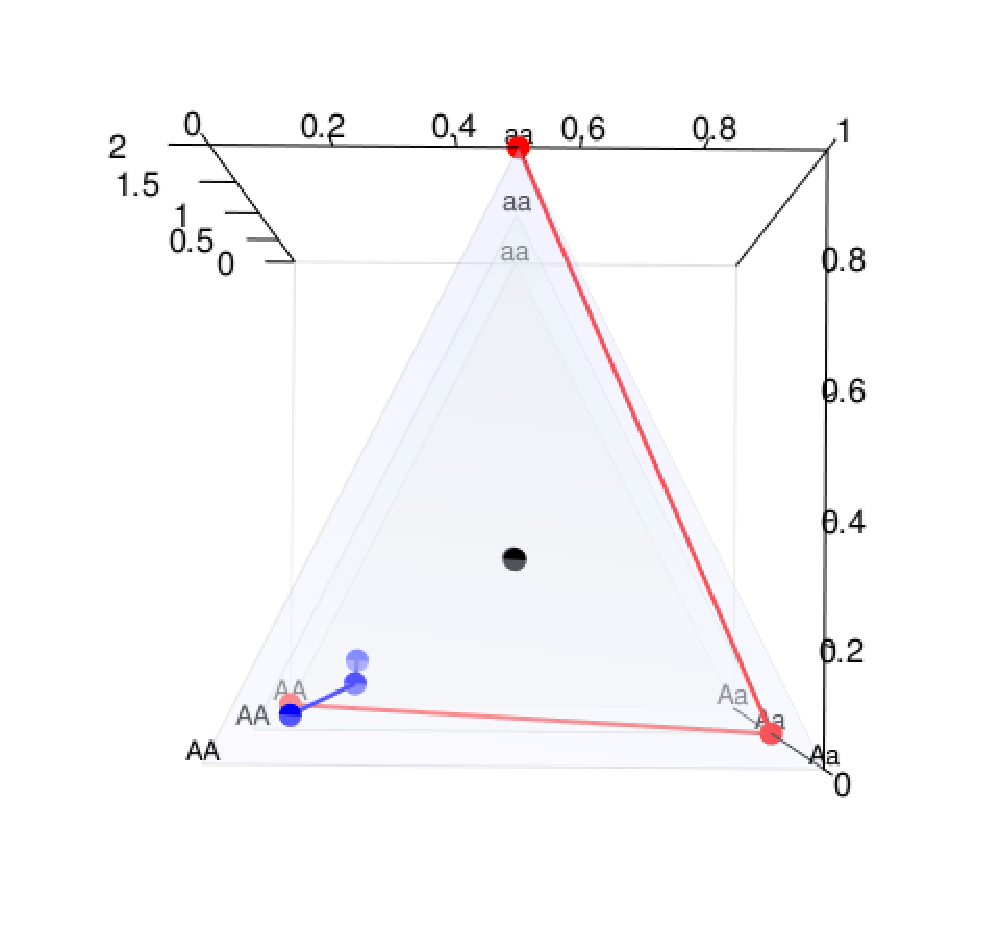}
        \caption{Top-down view}
        \label{fig:toytop}
    \end{subfigure}
    
   \caption{\small Plotting of three different toy examples. Panel (a) depicts the side view, where vertical alignment corresponds to geometric orthogonality across ternary layers, and panel (b) depicts the corresponding top-down projection used in the original two-dimensional plot. Each colored trajectory connects the conditional genotype distribution of a target SNP across the three genotype strata of a conditioning SNP, which is able to be plotted on stacked ternary planes corresponding to $C = 0, 1, 2$ due to them summing up to 1. The black trajectory represents an independent SNP; by having identical conditional probabilities for each strata connecting each point forms a line orthogonal to each ternary plane that collapses to a single point under top-down projection in (b). As such, a single point in the projected view corresponds geometrically to orthogonality and probabilistically to independence. The blue trajectory depicts mild dependence, where connecting the small shifts across strata with a line produces a smaller V-shape in the projected view. The red trajectory shows strong dependence, where the connection between the large shifts in conditional probabilities across the strata generates a pronounced V-shape. 
}
\label{fig:ct_sideview}
\end{figure}


A single dependent SNP produces one V-shaped trajectory under projection; when multiple SNPs are displayed relative to the same causal SNP, these individual V-shaped trajectories accumulate into larger geometric patterns. This is one of the advantages of crosstalk plots over many existing SNP visualization methods as they allow multiple SNPs to be layered within the same geometric framework. For example, a mosaic plot is limited to displaying a single contingency table at a time, so each SNP pair must be examined separately. In contrast, a crosstalk plot represents each SNP as a trajectory of conditional genotype probabilities across strata of a conditioning SNP, allowing many SNPs to be displayed simultaneously within one ternary coordinate system. This makes it possible to compare the behavior of all SNPs relative to the same conditioning structure and to detect common patterns that would otherwise be difficult to recognize across many separate plots. In Figure \ref{fig:chickenfeet}, we demonstrate this using the \textit{apolipoprotein E}  (APOE) locus, which has been central to Alzheimer disease research for decades. Variation at this locus is defined by two key SNPs, \texttt{rs429358} and \texttt{rs7412}, which give rise to three common APOE isoforms (APOE2, APOE3, APOE4) \citep{belloy2019apoe}. Since these variants jointly determine APOE isoforms, they provide a natural case study for Crosstalk plots. 

\begin{figure}[!]
    
    \centering
    \includegraphics[width=0.45\linewidth]{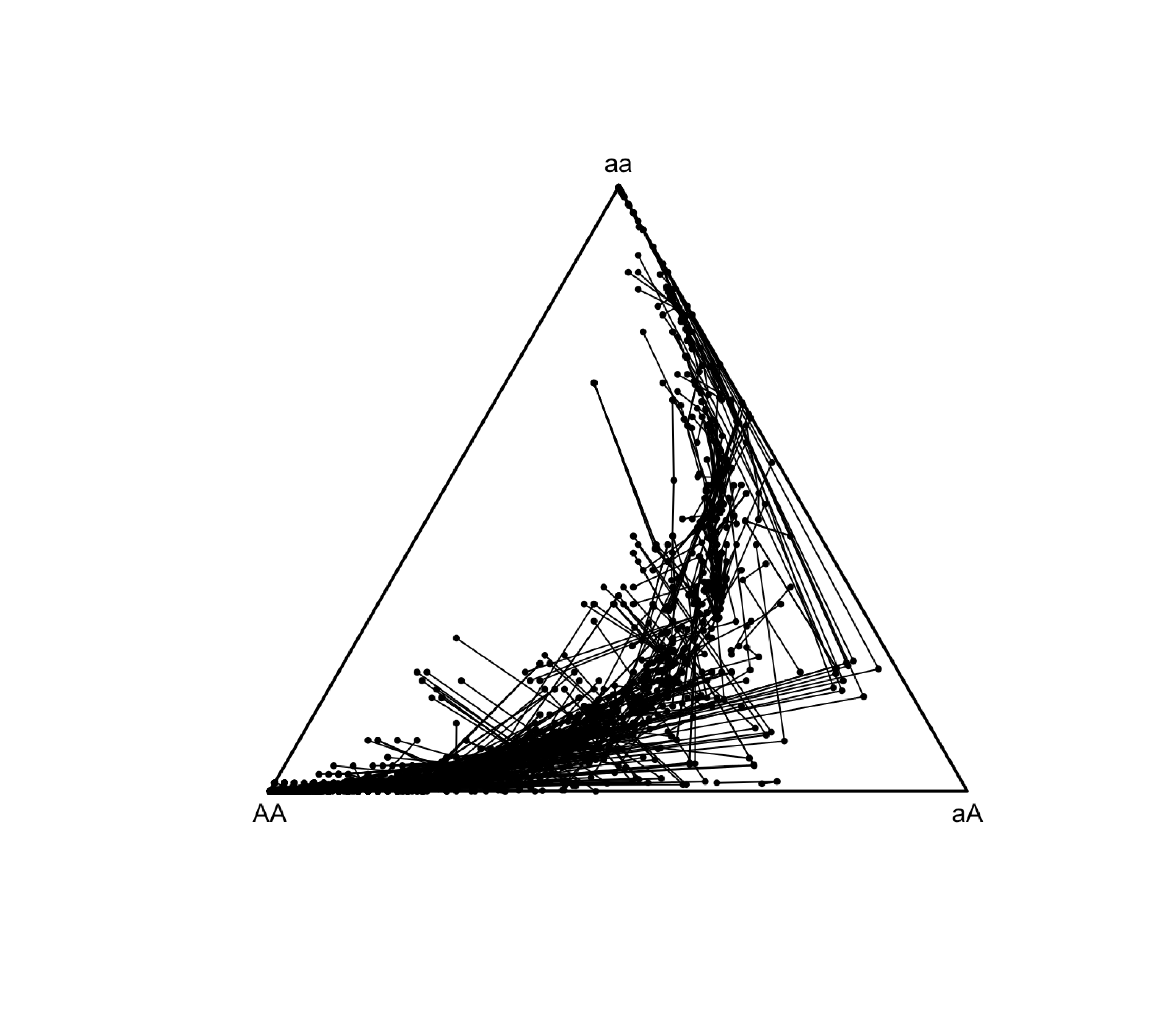}

    \caption{\justifying \small When many SNPs are layered together, shared dependence patterns become visually apparent that are not evident when examining individual SNP pairs. Using reference data from \citet{1000genomes2015}, we designate the \textit{apolipoprotein E} SNP \texttt{rs429358} as the causal SNP on which multiple SNPs are conditioned on. If neighboring markers were independent of the causal SNP, most projected trajectories would collapse to isolated points. However, the trajectories form repeated V-shapes, following a common directional trend and leaving a characteristic curved white space due to constraints imposed by the Hardy-Weinberg equilibrium (that the proportion of $Aa$ is $2pq$ with $P\{A\} = p$ and $P\{a\} = q$) together with dependence induced by linkage disequilibrium, indicating that the SNPs are not behaving independently but are shifting systematically across genotype strata. This geometric clustering has direct inferential implications in GWAS. Once one SNP is truly associated with an outcome, many dependent SNPs in the same region also become statistically associated because they carry partial statistical information about the same underlying signal as previously noted by \citet{DingEtAl(2021)}. As a result, large numbers of markers may appear significant even though they do not represent distinct biological discoveries. In such settings, false discovery rate (FDR) becomes difficult to interpret: the number of rejected hypotheses may become very large, not because many independent signals are present, but because one dependence structure generates repeated manifestations of the same association. Crosstalk plots make this phenomenon visible by showing that many significant SNPs often trace the same geometric pattern, reflecting underlying dependence structure rather than representing separate independent effects, underscoring why FDR can lose interpretability in GWAS.}
    \label{fig:chickenfeet}
\end{figure}


\section{Discussion}
\label{Sec:discussion}

While the applications of the conditionality principle and layering of error rates were focused on platform trials, there are numerous applications elsewhere. In fact, the effects of implicit conditioning are already evident in regulatory practice and drug labeling. A drug label is only constructed after an approval decision has been made. The act of labeling therefore occurs after conditioning on primary endpoint success, and the drug label represents a summary of evidence given that condition; additional claims must be supported by inference that is valid given primary end point success. Accordingly, regulators do not ask whether secondary endpoints perform marginally across all trials, but only in the paths that have achieved success on the primary endpoint. 

Similarly, the layering of error rates naturally arises because decisions are made both at the trial level and across the broader regulatory landscape. Within a specific study, incorrect conclusions about any endpoint may lead to erroneous claims, requiring strict control. We therefore want to ensure that the ERFw is controlled, i.e. the probability of making at least one incorrect inference among the endpoints in the trial does not exceed the prespecified level. At the higher layer, regulators also desire control of the error rate with respect to the approval of a drug compound, corresponding to IAR.

We also note that crosstalk plots have additional uses outside those listed in this paper. Such usages include visualizing linkage disequilibrium of a SNP.


\section{Conclusion}

The Conditionality Principle does not weaken frequentist guarantees; it rather aligns inference with the experiment that actually occurred.
\begin{enumerate}
	\item Among error rates defined by Tukey (1953), Error Rate Familywise (ERFw) applies within a clinical study, while Error Rate per Family (ERpF) applies \textit{across} studies.  
	\item An appropriate version of the Law of Large Numbers shows controlling ERFw \textit{within} a clinical study at $ 5\% $ in turn controls ERpF \textit{across} studies at 5-per-100, regardless of whether the studies are in Platform trials or not.  
	\item With ongoing regulatory practice controlling ERpF \textit{across} studies at 5-per-100 (via controlling the ERFw of each clinical study at $ 5\% $, the False Approval Rate (FAR) is being controlled at 5-per-100 (even in Platform trials), allowing only five incorrect approvals per 100 drug compounds that lack efficacy.  
	\item False Approval Rate (FAR), the proportion of drug compounds lacking efficacy incorrectly approved by regulators, should not be increased with increasing number of efficacious compounds approved (according to our understanding of the U.S. Food, Drug, \& Cosmetic Act and its Amendment).  
	\item Controlling the False Discovery Rate (FDR) \textit{across} studies at $ 5\% $ would make FAR more than 5-per-100, a change of (inter)national policies more than statistical that we believe would require harmonizsation/harmonization.  
\end{enumerate}  

In this paper, we examined the error rates within clinical studies, particularly focusing on platform trials. Our findings offer significant insights into the control of error rates and their implications across various study designs, especially in the context of regulatory and clinical settings.  

We demonstrated that controlling the ERFw within a single study at a threshold of $5\%$ effectively controls the ERpF across multiple studies at a rate of 5-per-100. This finding holds irrespective of the correlation between studies, including those grouped within platform trials. Furthermore, our analysis supports that current regulatory practices, which involve controlling the ERFw within a study, limit the False Approval Rate (FAR) to 5-per-100. This approach ensures that no more than five ineffective drug compounds receive approval per hundred evaluated, thus safeguarding public health while fostering innovation. 

 A significant advantage of a Platform trial is that multiple treatments are evaluated simultaneously against a shared control, which allows for fewer total participants, and greater allocation of participants to the active treatment within each regimen, relative to traditional designs with 1:1 randomization. However, the stability of regulatory FAR across Platform trials is of serious concern since shared controls introduce dependencies among treatment evaluations and therefore increases the variability of ERpf. As we demonstrated in our analysis, dependencies among treatment evaluations increases with the treatment/control randomization ratio. From patients perspective, a higher treatment/control randomization ratio will allow higher chance of receiving a new treatment compound. From an individual drug sponsor’s point of view, a higher treatment/control randomization ratio through participating platform trial allows more accurate estimate on the effects of new compounds. However, from the master protocol sponsor’s point of view,  lower the treatment/control randomization ratio allows increasing the sample size of the shared control, and therefore lessens the chance that an overachieving shared control sample causes all the new compounds have difficulty to demonstrate efficacy. Additionally, from the regulators’ point of view, stability of ERpF control is affected by both stability of the sample mean of the shared control as well as correlation among the pivotal statistics. A lower randomization ratio increases sample size of the shared control, which lessens the chance of an under-performing shared control sample giving all the new compounds an easy pass. A lower randomization ratio also decreases correlation among the pivotal statistics, which makes ERpF control more stable. To balance the different perspective from different stake holder, our Confident RandomSizer app offers a tool to identify optimal treatment/control randomization ratio given a fixed sample size. 

In conclusion, maintaining stringent control of error rates within clinical trials is crucial for ensuring the reliability of study outcomes and the safety of approved treatments. As clinical trial designs become increasingly complex, particularly with the rise of platform trials, the statistical community and regulatory bodies must continuously adapt and refine their methodologies to address these challenges effectively.
\nocite{MassGeneral(2024)}
\nocite{FDA(2023)}

\section{Conflict of Interest}

X.C. and E.Ou declare no competing interests. Y. L. is an employee and shareholder at Nektar Therapeutics, San Francisco, CA, USA. J.Y.S and H. T. are employees and shareholders of BeOne Medicines (San Mateo, CA, USA). B.W. is an employee at IO Biotech (Copenhagen, Denmark). J.C. H. is a professor emeritus of the Ohio State University and served as a consultant to BeOne Medicines and Eli Lilly.
\section{Glossary}
\subsection{Terms}
\begin{itemize}
    \item \textbf{ALS}: Amyotrophic lateral sclerosis (ALS)/Lou Gehrig’s disease; a fatal progressive neuromuscular disease
    \item \textbf{ERpF}: Error Rate per Family; measures errors across studies
    \item \textbf{ERFw}: Error Rate Family wise; error within each study
    \item \textbf{FDR}: False Discovery Rate; expected proportion of rejections that are incorrect
    \item \textbf{FWER}: Family Wise Error Rate; probability of type I error in a family
        \item \textbf{IAR}: Incorrect Approval Rate; proportion of drug compounds that lack efficacy incorrectly approved
\end{itemize}

\subsection{Notation}
\begin{itemize}
    \item $\mu_{SC}$: Shared control true sample mean
    \item $\hatmusc$: Shared control sample mean estimate
    \item $X$: Indicator that type I error is committed as a result of $i$th study 
    \item $V$: The number of Type I errors committed in a family
    \item $Z$: Indicator that $i$th confidence interval erroneously shows efficacy
    \item $V^\star$:  The number of treatments with their confidence intervals for efficacy not covering their true parameter values given $n$ confidence intervals
    \item $SE$: Standard error
    \begin{itemize}
        \item Subscripts denote standard error of respective term
    \end{itemize}
    \item $k$: Number of treatment arms
    \item $\ell$: Number of non-efficatious treatment arms
    \item $n$: Maximum number of treatment arms erroneously approved
    \item $T^\star_i$: Statistic for comparing $k$ new compounds against shared control
    \[T^\star_i = \dfrac{\hat{\mu}_{i}-\hat{\mu}_{SC}}{SE_{\hat{\mu}_{i}-\hat{\mu}_{SC}}}\]
    \item $R$: the number of correct rejection of false nulls
    \item $n_{SC}$: sample size for shared control
    \item $n_i$: sample size for $i$th treatment
    \item $\lambda_i$: Have that \[\lambda_i = \left(1+\frac{n_{SC}}{n_{i}}\right)^{-\frac{1}{2}}\]
    \item $\rho_{ij}$: Correlation between sample statistics $T_i$ and $T_j$
    \begin{itemize}
        \item Note: Can be written as a product of $\lambda_i$ and $\lambda_j$.
    \end{itemize}
    
\end{itemize}




\bibliography{refsPlatform}

\begin{table}[tbh]
	\centering 
	\begin{tabular}{|c|c| c c c|}
		\hline
		$k$ & $E(V^\star$) &	\multicolumn{3}{|c|}{$ StdDev(V^\star) $ } \\ 
        & (Unconditional ERpF) & Independent studies & Hypothetical & Platform $k:1$ randomization \\
		 &  & $ \rho_{ij} = 0$ & $ \rho_{ij} = 0.3$ & $\rho_{ij} = 0.5$ \\
		\hline
		5 & 0.25 & 0.49 & 0.57 & 0.66 \\
		\hline
		10 & 0.5 & 0.69 & 0.94 & 1.16 \\
		\hline
		20 & 1 & 0.97 & 1.65 & 2.15 \\
		\hline
		40 & 2 & 1.38 & 3.02 & 4.12 \\
		\hline
	\end{tabular}
    \caption{Mean and standard deviation of $V^\star$ for different family sizes $k$ under varying correlation structures among test statistics $T_i^\star$. Independence corresponds to $\rho_{ij}=0$, while the correlated scenarios assume a common correlation $\rho_{ij}=0.3$ or $\rho_{ij}=0.5$ among all pairs of treatment comparisons. The case $\rho_{ij}=0.5$ corresponds to a platform trial with $k$ experimental compounds using $k:1$ compound-to-placebo randomization, which induces correlation among $T_i^\star$ through the shared control arm. The expected value $E(V^\star)$ corresponds to the unconditional ERpF level with ERpF set to $5\%\times k$.}\label{tab:VariabilityVstar}
\end{table}


\begin{table}
    \centering
    \begin{tabular}{|c|c|c|c|c|}
    \hline
        $k$ & A & B & C & D\\
        \hline
	A  & \textbf{122} & 122 & 164 & 103 \\
 	\hline
        B & 122 & \textbf{120} & 122 & 103 \\
        \hline
        C & 164 & 122 & \textbf{120} & 103 \\
        \hline
        D & 103 & 103 & 103 & \textbf{75} \\
        \hline
    \end{tabular}
    \caption{Control Sharing Matrix (CSM) for the HEALEY ALS Platform Trial}
    \label{tab:CSM}
\end{table}

\begin{table}
    \centering
    \begin{tabular}{|c|c|c|c|c|}
    \hline
          & A & B & C & D \\
        \hline
       A  & 1 & 0.498 & 0.425 & 0.478 \\
       \hline
       B  & 0.498 & 1 & 0.496 & 0.476 \\
       \hline
       C  & 0.425 & 0.496 & 1 & 0.476 \\
       \hline
       D  & 0.478 & 0.476 & 0.476 & 1 \\
       \hline
    \end{tabular}
    \caption{Correlations matrix between $\hat{\mu}_{i}-\hat{\mu}_{SC_{ij}}$ and $\hat{\mu}_{j}-\hat{\mu}_{SC_{ij}}$}for the HEALEY ALS Platform Trial
    \label{tab:cor_mat}
\end{table}

\begin{table}[tbh]
    \centering
    \begin{tabular}{|c|c|c|c|c|}
    \hline
          & A & B & C & D \\
    \hline
       A  & 0.0475 & 0.0093 & 0.0083 & 0.0087 \\
    \hline
       B  & 0.0093 & 0.0475 & 0.0092 & 0.0096 \\
    \hline
       C  & 0.0083  & 0.0092 & 0.0475 & 0.0086 \\
    \hline
       D  & 0.0087 & 0.0096 & 0.0086 & 0.0475 \\
    \hline
    \end{tabular}
    \caption{Matrix of variances of $Z_i$ (diagonal) and covariances between $Z_i$ and $Z_j$ (off-diagonal) with ERpF set at $ 5\% \times 4 = 0.2$ for Regimens ABCD of the ALS Platform trial.}
    \label{tab:CorrelationsZ}
\end{table}
\end{document}